\documentclass[11pt,a4paper]{article}
\usepackage{amsmath}
\usepackage[T1]{fontenc}
\usepackage[utf8]{inputenc}
\usepackage{amsfonts}
\usepackage{amsthm}
\usepackage{amstext}
\usepackage{amssymb}
\usepackage{bbold}
\usepackage{hyperref}
\usepackage{url}
\newcommand{\rr}{\mathbb{R}}
\newcommand{\rd}{\rr_d}

\newtheorem{theorem}{Theorem}

\newtheorem{lemma}[theorem]{Lemma}
\newtheorem*{lemma*}{Lemma}
\newtheorem{proposition}[theorem]{Proposition}
\newtheorem{corollary}[theorem]{Corollary}
\newtheorem{remark}[theorem]{Remark}

\newtheorem*{open*}{Open~question}
\newtheorem{definition}[theorem]{Definition}

\begin{document}

\title{Lower Bounds by Birkhoff Interpolation\thanks{This work was
    supported by ANR  project CompA (project number:
    ANR-13-BS02-0001-01).
Authors' email addresses: [ignacio.garcia-marco, pascal.koiran]@ens-lyon.fr.}}

\author{Ignacio Garcia-Marco and Pascal Koiran\\
LIP\thanks{UMR 5668 ENS Lyon - CNRS - UCBL - INRIA, Universit\'e de
  Lyon}, Ecole Normale Supérieure de Lyon}

\date{July 7, 2015}

\maketitle

\begin{abstract}
In this paper we  give lower bounds for the representation of
real univariate polynomials as sums of powers of degree 1 polynomials.
We present two families of polynomials of degree $d$ such that the number of powers
that are required in such a representation must be at least of order $d$. 
This is clearly optimal up to a constant factor.
Previous lower bounds for this problem were only of order $\Omega(\sqrt{d})$, and were obtained from arguments 
based on Wronskian determinants and "shifted derivatives."
We obtain this improvement thanks to a new lower bound method based on Birkhoff interpolation 
(also known as "lacunary polynomial interpolation").
\end{abstract}

\section{Introduction}

In this paper we obtain lower bounds for the representation of
a univariate polynomial $f \in \rr[X]$ of degree $d$ under the form:
\begin{equation} \label{model}
f(x)=\sum_{i=1}^l \beta_i (x+y_i)^{e_i}
\end{equation}
where the $y_i$ are real constants and the exponents $e_i$ nonnegative
integers.

We give two families of polynomials such that the number $l$ of terms 
required in such a representation must be at least of order $d$. 
This is clearly optimal up to a constant factor.
Previous lower bounds for this problem~\cite{KKPS} were only of order $\Omega(\sqrt{d})$.
The polynomials in our first family are of the form 
$H_1(x)=\sum_{i=1}^k \alpha_i (x+x_i)^d$ with all $\alpha_i$ nonzero
and the $x_i$'s distinct.
We show that that they require at least $l \geq k$ terms whenever $k
\leq (d+2)/4$. In particular, for $k=(d+2)/4$ we obtain $l=k=(d+2)/4$ as a lower
bound.
The polynomials in our second family are of the form 
$H_2(x)=(x+1)^{d+1}-x^{d+1}$ and we show that they require more than
$(d-1)/2$ terms. This improves the lower bound for $H_1$ by a factor
of 2, but this second lower bound applies only when the exponents
$e_i$ are required to be bounded by $d$ (obviously, if larger
exponents are allowed we only need two terms to represent $H_2$).
It is easily shown that every polynomial of degree $d$ can be
represented with $\lceil(d+1)/2 \rceil$ terms.
This implies that of all polynomials of degree
$d$, $H_2$ is essentially (up to a small additive constant) the
hardest one.

Our lower bound results are specific to polynomials with real coefficients.
It would be interesting to obtain similar lower bounds for other
fields, e.g., finite fields or the field of complex numbers.
As an intermediate step toward our lower bound theorems, we obtain a result on the linear independence of polynomials which may be of independent interest.
\begin{theorem} \label{independence}
Let $f_1,\ldots,f_k \in \rr[X]$ be $k$ distinct polynomials of the form
$f_i(x)=(x+a_i)^{e_i}$. 
Let us denote by $n_j$ the number of polynomials of degree less than $j$ in this family.

If
$n_1 \leq 1$ and $n_j+n_{j-1}
\leq j$ for all $j$, the family $(f_i)$  is linearly independent.
\end{theorem}
We will see later (in Section~\ref{lb_section}, Remark~\ref{opt}) that this theorem is optimal up to a small
additive constant when $d$ is even, and exactly optimal when $d$ is odd.

\subsection*{Motivation and connection to previous work}

Lower bounds for the representation of univariate polynomials as sums 
of powers of {\em low degree} polynomials were recently obtained in~\cite{KKPS}.
We continue this line of work by focusing on
powers of {\em degree one} polynomials. This problem is still
challenging because the exponents $e_i$ may be different from
$d=\deg(f)$, and may be possibly larger than $d$. 
The lower bounds obtained in~\cite{KKPS} are of order
$\Omega(\sqrt{d})$.
We obtain $\Omega(d)$ lower bounds with 
a new method
based on polynomial interpolation (more on this below).

The work in~\cite{KKPS} and in the present paper is motivated by
recent progress in arithmetic circuit complexity. It was shown
that strong enough lower bounds for circuits of depth four~\cite{AgraVinay08,Koi12,Tavenas} or
even depth three~\cite{GuptaKKS13,Tavenas}
 would yield a separation of Valiant's~\cite{Valiant79} algebraic
complexity classes VP and VNP.
Moreover, lower bounds for such circuits were obtained thanks to the
introduction by Neeraj Kayal of the method of {\em shifted partial
  derivatives}, see
e.g.~\cite{Kayal12,fournier2014lower,GKKS13,kayal2014exponential,kayal2014super,kumar2014power,kumar2014limits}.
Some of these lower bounds seem to come close to separating VP from
VNP, but there is evidence that the method of shifted derivatives
by itself will not be sufficient to achieve this goal. It is therefore
desirable to develop new lower bounds methods. 
We view the models studied in~\cite{KKPS} and in the present paper as "test beds" for the development of such methods in a fairly simple setting.
We note also that (as explained above) strong lower bounds in slightly more general models would imply a separation of VP from VNP. 
Indeed, if  the affine function $x+y_i$ in~(\ref{model}) are replaced
by a multivariate affine functions we obtain the model of "depth 3
powering arithmetic circuits.'' In general depth 3 arithmetic
circuits, instead of  powers of affine functions we have products of
(possibly distinct) affine functions. We note that the depth reduction
result of~\cite{GuptaKKS13} yields circuits where the number of
factors in such products can be  much larger than the degree of the polynomial represented by the circuit. It is therefore quite natural to allow exponents $e_i > d$ in~(\ref{model}). Likewise, the model studied in~\cite{KKPS} is close to depth 4 arithmetic circuits, see~\cite{KKPS} for details.

\subsection*{Birkhoff interpolation}

As mentioned above, our results are based on polynomial interpolation and more precisely on Birkhoff interpolation (also known as "lacunary interpolation"). The most basic form of polynomial interpolation is Lagrange interpolation. In a typical Lagrange interpolation problem, one may have to find a polynomial $g$ of degree at most 2 satisfying the 3 constraints $g(-1)=1$, $g(0)=4$, $g(1)=3$.
At a slightly higher level of generality we find Hermite
interpolation, where at each point we must interpolate not only values
of $g$ but also the values of its first few derivatives. As an
example, we may have to find a polynomial $g$ of degree 3 satisfying
the 4 constraints  $g(0)=1$, $g(1)=0, g'(1)=-1, g''(1)=2$. 
Birkhoff interpolation is even more general as there may be ``holes''
in the sequence of derivatives to be interpolated at each point.
An example of such a problem is: $g(0)=0$, $g'(1)=0$, $g(2)=g"(2)=0$.
We have set the right hand side of all constraints to 0 because the
interpolation problems that we need to handle in this paper all turn
out to be of that form (in general, one may naturally allow nonzero
values). Our interest is in the existence of a nonzero polynomial of
degree at most $d$ satisfying the constraints, and more generally in
the dimension of the solution space. In fact, we need to know whether
it has the dimension
that one would expect by naively couting the number of constraints.
 Contrary to Lagrange or
Hermite interpolation in one variable, where the existence of a nonzero solution can be easily decided by
comparing the number of constraints to $d+1$ (the number of
coefficients of~$g$), this is a nontrivial problem and a rich theory
was developed to address it~\cite{lorentz84}.
Results of the real (as opposed to complex) theory of Birkhoff interpolation 
 turn out to be very well suited to our lower bound problems. This is the
 reason why we work with real polynomials in this paper.

\subsection*{The Waring problem}

Any homogenous (multivariate) polynomial $f$ can be written as a sum of
powers of linear forms. In the Waring problem for polynomials one
attempts to determine the smallest possible number of powers in such a
representation. This number is called the {\em Waring rank} of $f$.
Obtaining lower bounds from results in polynomial interpolation
seems to be a new method in complexity theory, but it may not come as
a surprise to experts on the Waring problem. Indeed, a major result in
this area, the Alexander-Hirschowitz theorem (\cite{alexander95}, see~\cite{brambilla08} for a survey), is usually stated as a
result on (multivariate, Hermite) polynomial interpolation.
Classical work on the Waring problem was focused on the Waring rank of
{\em generic polynomials}, and this question was completely answered by 
 Alexander and Hirschowitz.
The focus on generic polynomials is in sharp contrast with complexity theory, where a main goal is to prove lower bounds on the complexity of {\em explicit} polynomials (or of explicit Boolean functions in Boolean complexity).
A few  recent papers~\cite{landsberg2010,carlini2012} have begun to investigate the Waring rank of specific (or explicit, in computer science parlance) polynomials such as monomials, sums of coprime monomials, the permanent and the determinant. We expect that more connections between lower bounds in algebraic complexity, polynomial interpolation and the Waring problem will be uncovered in the future.

\subsection*{Organization of the paper}

In Section~\ref{translation} we begin a study of the linear
independence of polynomials of the form $(x+y_i)^{e_i}$. We show that
this problem can be translated into a problem of Birkhoff
interpolation, and in fact we show that Birkhoff interpolation and
linear independence are dual problems. In Section~\ref{matrices} we
present the notions and results on Birkhoff interpolation that are
needed for this paper, and we use them to prove
Theorem~\ref{independence}. We build on this result to prove our lower
bound results in Section~\ref{lb_section}, and we discuss their
optimality. 
The lower bound problem studied in this paper is over the field of real numbers.
In Section~\ref{fields} we briefly discuss the situation in other fields and in particular the field of complex numbers.
Finally, we give an illustration of our methods in the appendix by completely working out a small example.

\section{From linear independence to polynomial interpolation} \label{translation}

There is a clear connection between lower bounds for representations
of polynomials under form~(\ref{model}) and linear independence.
Indeed, proving a lower bound for a polynomial  $f$ amounts to showing that $f$ is linearly independent from $(x+y_1)^{e_1},\ldots,(x+y_l)^{e_l}$ for some $l$ and for any sequence of $l$ pairs  $(y_1,e_1),\ldots,(y_l,e_l)$. Moreover, if the "hard polynomial" $f$ is itself presented as a sum of powers of degree 1 polynomials (which is the case in this paper), we can obtain a lower bound for $f$ from linear independence results for such powers. This motivates the following study.

Let us denote by $\rd[X]$ the linear subspace of $\rr[X]$ made of
polynomials of degree at most $d$, and by $g^{(k)}$ th $k$-th order
derivative of a polynomial $g$.

\begin{proposition} \label{dual}
Let $f_1,\ldots,f_k \in \rd[X]$ be $k$ distinct polynomials of the form
$f_i(x)=(x+a_i)^{e_i}$.
The family $(f_i)_{1 \leq i \leq k}$ is linearly independent if and
only if
$$\dim \{g \in \rd[X];\ g^{(d-e_i)}(a_i) =0 \text{ \rm for all } i\} = d+1-k.$$
\end{proposition}
Let $V$ be the subspace of $\rd[X]$ spanned by the $f_i$. The
orthogonal $V^{\perp}$ of $V$ is the space of linear forms $\phi \in
\rd[X]^*$ such that $\langle \phi,f \rangle = 0$ for all $f \in V$.
We will use the fact that $\dim V^{\perp}  = d+1 - \dim V$. We will
identify $\rd[X]$ with its dual $\rd[X]^*$ via the symmetric
bilinear form
$$\langle g,f \rangle = \sum_{k=0}^d \frac{f_k g_{d-k}}{{d \choose k}}.$$
This is reminiscent of Weyl's unitarily invariant inner
product (see e.g. chapter 16 of~\cite{burgisser2013} for a recent exposition)
but we provide here a self-contained treatment.
Poposition~\ref{dual} follows immediately from the next lemma:

\begin{lemma}
The orthogonal $f_i^{\perp}$ of $f_i$ is equal to $
\{g \in \rd[X];\ g^{(d-e_i)}(a_i) =0\}$.
\end{lemma}

\begin{proof}
We begin with the case $e_i=d$. We need to show that for a polynomial
$f(x)=(x+a)^d$, $\langle g,f \rangle =0$ iff $g(a)=0$.
This follows from the definition of $\langle g,f \rangle$ since by
expanding $(x+a)^d$ in powers of $x$ we have
\begin{equation} \label{lagrange}
\langle g,(x+a)^d \rangle = \sum_{k=0}^d g_{d-k}a^{d-k}=g(a).
\end{equation}

Consider now the general case $f(x)=(x+a)^{d-k}$ where $k \geq 0$.
We will show that
\begin{equation} \label{birkhoff}
g^{(k)}(a)= \frac{d!}{(d-k)!} \langle g,f \rangle,
\end{equation}
thereby completing the proof of the lemma.
In order to obtain~(\ref{birkhoff}) from~(\ref{lagrange})
we introduce a new variable $\epsilon$ and expand in two different ways
$\langle g,(x+a+\epsilon)^d \rangle$ in powers of $\epsilon$.
From~(\ref{lagrange}) we have
\begin{equation} \label{expansion1}
\langle g,(x+a+\epsilon)^d \rangle = g(a+\epsilon) =
\sum_{k=0}^d \frac{g^{(k)}(a)}{k!} \epsilon^k.
\end{equation}
On the other hand, since
$(x+a+\epsilon)^d =
\sum_{k=0}^d {d \choose k}\epsilon^k (x+a)^{d-k}$
we have from bilinearity
\begin{equation} \label{expansion2}
\langle g,(x+a+\epsilon)^d \rangle = \sum_{k=0}^d {d \choose k}
\langle g,(x+a)^{d-k} \rangle \epsilon^k.
\end{equation}
Comparing~(\ref{expansion1}) and~(\ref{expansion2}) shows that
$\frac{g^{(k)}(a)}{k!} = {d \choose k}
\langle g,(x+a)^{d-k} \rangle$.
\end{proof}

Since $\rd[X]$ has dimension $d+1$ we must have $k \leq d+1$ for the
$f_i$ to be linearly independent. More generally, let $n_j$ be the
number of $f_i$'s which are of degree less than $j$.
Again, for the $f_i$ to be linearly independent we must have $n_j \leq
j$ for all $j=1,\ldots,d+1$.
The polynomial
identity $(x+1)^2-(x-1)^2-4x=0$ shows that the converse is not true,
but Theorem~\ref{independence} from the introduction shows that a weak converse holds true.
We will use Proposition~\ref{dual} to prove this theorem at the end of the next section.

\section{Interpolation matrices} \label{matrices}

In Birkhoff interpolation we look for a polynomial $g \in \rd[X]$
satisfying a system of linear equations of the form
\begin{equation} \label{birkint}
g^{(k)}(x_i)=c_{i,k}.
\end{equation}
The system may be lacunary, i.e., we may not have an equation in the
system for every value of $i$ and $k$. We set $e_{i,k}=1$ if such an
equation appears, and $e_{i,k}=0$ otherwise. We arrange this
combinatorial data in an {\em interpolation matrix}
$E=(e_{i,k})_{1 \leq i \leq m, 0 \leq k \leq n}$. We assume that the
{\em knots} $x_1,\ldots,x_m$ are distinct.
It is usually assumed~\cite{lorentz84}
that $|E|=\sum_{i,k} e_{i,k}$, the number of 1's in $E$, is equal to
$d+1$ (the number of coefficients of $g$). Here we will only assume
that $|E| \leq d+1$. We can also assume that $n \leq d$ since $g^{(k)}=0$ for
$k>d$. In the sequel we will in fact assume that $n=d$: this condition
can be enforced by adding empty columns to $E$ if necessary.

Let $X=\{x_1,\ldots,x_m\}$ be the set of knots. When $|E|=d+1$, the
pair $(E,X)$ is said to be {\em regular} if~(\ref{birkint}) has a
unique solution for any choice of the $c_{i,k}$.
Finding necessary or sufficient conditions for regularity
has been a major topic in Birkhoff interpolation~\cite{lorentz84}.
For $|E| \leq d+1$, we may expect~(\ref{birkint}) to have a set of
solutions of dimension $d+1-|E|$. We therefore extend the definition
of regularity to this case as follows.
\begin{definition}
The pair $(E,X)$ is regular if for any choice of the $c_{i,k}$
the set of solutions
of~(\ref{birkint}) is an affine subspace of dimension $d+1-|E|$.
\end{definition}
Note in particular that the set of solutions is nonempty
since $|E| \leq d+1$.

Basic linear algebra provides a link between regularity
for different values of $|E|$.
\begin{proposition} \label{subset}
Let $E$ be an $m \times (d+1)$ interpolation matrix. For an
interpolation matrix $F$ of the same format, we write $F \subseteq E$
if $e_{i,k}=0$ implies $f_{i,k}=0$ (i.e., the set of 1's of $F$ is
included in the set of 1's of $E$).

If the pair $(E,X)$ is regular and $F \subseteq E$ then $(F,X)$ is
regular as well.
\end{proposition}
\begin{proof}
Consider the interpolation problem:
$$g^{(k)}(x_i)=c_{i,k} \text{ for } f_{i,k}=1.$$
The set of solutions ${\cal F} \subseteq \rd[X]$ is an affine subspace
which is either empty or of dimension at least $d+1-|F|$.
It cannot be empty since by adding $|E|-|F|$ constraints we can obtain
an interpolation problem with a solution space of dimension
$d+1-|E| \geq 0$. For the same reason, it cannot be of dimension
$d+2-|F|$ or more. In this case, by adding $|E|-|F|$ constraints we
would obtain an interpolation problem with a solution space of dimension
at least $(d+2-|F|)-(|E|-|F|) = d+2-|E|$. This is impossible since
$(E,X)$ is regular.
\end{proof}
Another somewhat more succint way of phrasing the above proof is to
consider the matrix of the linear system defining the affine
subset~$\cal F$. Anticipating on Section~\ref{lb_section}, let us
denote this matrix by $A(E,X)$. The pair $(E,X)$ is regular
iff $A(E,X)$ if of full row rank. The rows of $A(F,X)$ are also rows
of $A(E,X)$, so $A(F,X)$ must be of full row rank if $A(E,X)$ is.

For an interpolation matrix,
the notions of {\em regularity} and {\em order regularity} are
classicaly defined~\cite{lorentz84} in the case  $|E| = d+1$, but the extension to the
general case $|E| \leq d+1$ is straightforward:
\begin{definition}
The interpolation matrix $E$ is regular if $(E,X)$ is regular for
any choice of $m$ knots $x_1,\ldots,x_m$. It is order regular if
$(E,X)$ is regular for any choice of $m$ ordered knots $x_1 < x_2
\ldots < x_m$.
\end{definition}
As an immediate corollary of Proposition~\ref{subset} we have:
\begin{corollary} \label{subsetcor}
Let $E,F$ be two interpolation matrices with $F \subseteq E$.
If $E$ is regular (respectively, order regular) then $F$ is
also regular (respectively, order regular).
\end{corollary}

We will give in Theorem~\ref{order} a sufficient condition for order
regularity, but we first need some additional definitions.
Say that an interpolation matrix $E$ satisfies the {\em upper
  P\'olya condition} if for $r=1,\ldots,d+1$ there are at most~$r$ 1's
in the last $r$ columns of $E$.
If $|E|=d+1$ this is equivalent to the {\em P\'olya condition:}
 there are at least $r$ 1's
in the first $r$ columns of $E$ for $r=1,\ldots,d+1$.

Consider a row of an interpolation matrix $E$. By {\em sequence} we
mean a maximal sequence of consecutive 1's in this row. A sequence
containing an odd number of 1's is naturally called an {\em odd sequence}.
A sequence of the $i$th row is {\em supported} if there are 1's in $E$
both to the northwest and southwest of the first element of the row.
More precisely, if $(i,k)$ is the position of the first 1 of the
sequence, $E$ should contain 1's in positions $(i_1,k_1)$ and
$(i_2,k_2)$ where $i_1 < i < i_2$, $k_1 < k$ and $k_2 < k$.
The following important result (Theorem~1.5 in~\cite{lorentz84}) is
due to Atkinson and Sharma~\cite{atkinson69}.
\begin{theorem} \label{order}
Let $E$ be an $m \times (d+1)$ interpolation matrix with $|E|=d+1$.
If $E$ satisfies the P\'olya condition and contains no odd supported
sequence then $E$ is order regular.
\end{theorem}
As an example, the interpolation problem corresponding to the
polynomial identity $(x+1)^2-(x-1)^2-4x=0$
is:
$$g(-1)=0,\ g'(0)=0,\ g(1)=0$$
where $g \in {\rr}_2[X]$. It admits $g(x)=x^2-1$ as a nontrivial solution.
The corresponding interpolation matrix
 $$\left(\begin{array}{ccc}
1 & 0 & 0\\
0 & 1 & 0\\
1 & 0 & 0
\end{array}\right)$$
satisfies the P\'olya condition but contains an odd supported sequence in its second row.
\begin{corollary} \label{ordercor}
Let $E$ be an $m \times (d+1)$ interpolation matrix with $|E|=d+1$.
If $E$ satisfies the P\'olya condition, then:
\begin{itemize}
\item[(i)] if every odd
sequence of $E$ belongs to the first row, to the last row, or begins in
the first column then $E$ is order regular.
\item[(ii)] if every odd
sequence of $E$ begins in
the first column then $E$ is regular.
\end{itemize}
\end{corollary}
\begin{proof}
Part (i) follows from the fact that a sequence which belongs to the
first row, to the last row, or begins in the first column cannot be
supported.

For part (ii), assume that every odd sequence of $E$ begins in
the first column and fix $m$ distinct nodes $x_1,\ldots,x_m$.
By reordering the $x_i$'s we
obtain an increasing sequence $x'_1 < x'_2 < \cdots < x'_m$.
Applying the same permutation on the rows of $E$, we obtain an
interpolation matrix $E'$; clearly, the pair $(E,X)$ is regular if and
only if $(E',X')$ is. The latter pair is regular because $X'$ is ordered
and (by part (i)) $E'$ is order regular.
\end{proof}
We can now prove the main result of this section.
\begin{theorem} \label{regular}
Let $F$ be an $m \times (d+1)$ interpolation matrix. We denote by
$N_r$ the number of 1's in the last $r$ columns of $F$. If $N_1 \leq
1$ and $N_r + N_{r-1} \leq r$ for $r=2,\ldots, d+1$ then $F$ is regular.
\end{theorem}
Note that the conditions on $N_r$ are a strengthening of the upper P\'olya
condition $N_r \leq r$.
\begin{proof}
We will add 1's to $F$ so as to obtain a matrix
$E$ satisfying the hypothesis of Corollary~\ref{ordercor}, part~(ii).
Corollary~\ref{subsetcor} will then imply that $F$ is regular.

In order to construct $E$ we proceed as follows. First, for every
odd sequence of $F$ which does not begin in the first column we add
a 1 in the cell immediately to the left of its first 1. All odd
sequences of the resulting matrix $F'$ begin in the first column.
Moreover, we have added at most $N_{r-1}$ 1's in the last $r$
columns of $F$ (note that we can add exactly $N_{r-1}$ 1's when the
last $r-1$ columns contain $N_{r-1}$ sequences of length 1). Since
$N_1 \leq 1$ and $N_r + N_{r-1} \leq r$, $F'$ satisfies the upper
P\'olya condition. If $|F'|=d+1$ we set $E=F'$. This matrix
satisfies the P\'olya condition and its odd sequences all begin in
the first column, so we can indeed apply Corollary~\ref{ordercor} to
get that $E$ is regular and. Since  $F \subseteq E$, by Corollary
\ref{subsetcor}, we conclude that $F$ is also regular.

If $|F'| < d+1$ we need to add more 1's. It suffices to add $d + 1 -
|F'|$ new rows to $F'$ with a $1$ in the first column and $0$'s
everywhere else. Denoting by $E$ the resulting matrix we clearly
have that $E$ satisfies the P\'olya condition, $|E| = d+1$ and its
odd sequences begin in the first column, so Corollary~\ref{ordercor}
and Corollary~\ref{subsetcor} apply here to conclude that $F$ is
regular. Note that $E$ and $F$ do not have the same format since $E$
has more rows, but we can apply Corollary~\ref{subsetcor} if we first
expand $F$ with $d+1-|F'|$ empty rows.
\end{proof}

\subsection*{Proof of Theorem~\ref{independence}}

At this point we have enough knowledge of Birkhoff interpolation to prove Theorem~\ref{independence}.
In view of Proposition~\ref{dual} we need to show that
the interpolation problem
$$g^{(d-e_i)}(a_i) =0 \text{ \rm for } i=1,\ldots,k$$
has a solution space of dimension $d+1-k$.
Let $F$ be the corresponding interpolation matrix.
This matrix contains $d+1-k$ 1's and is of size $m \times (d+1)$ for
some $m \leq k$ (we have $m=k$ only when the $a_i$'s are all
distinct).
The hypothesis on the $n_j$'s implies that $F$ satisfies the
hypothesis of Theorem~\ref{regular}, and the result follows from the
regularity of $F$.

\section{Lower bounds} \label{lb_section}

System~(\ref{birkint}) is a linear system of equations in the coefficients of $g$. Following~\cite{lorentz84}, to set up this system it is convenient to work in the basis $(x^j/j!)_{0 \leq j \leq d}$ instead of the standard basis $(x^j)_{0 \leq j \leq d}$. We denote by $A(E,X)$ the matrix of the system in that basis, where as in the previous section $E$ denotes the corresponding interpolation matrix
and $X$ the set of knots. As already pointed out after Proposition~\ref{subset}, the pair $(E,X)$ is regular if and
only if $A(E,X)$ is of rank $|E|$. In our chosen basis, an
interpolation constraint of the form~(\ref{birkint}) reads:
$$\sum_{j=0}^d \frac{x_i^{j-k}}{(j-k)!} g_j = c_{i,k}$$
where the coefficients $g_0,\ldots,g_d$ are the unknowns and we
choose as in~\cite{lorentz84} to interpret $1/r!$ as 0 for $r<0$.
\begin{proposition} \label{split}
Consider a pair $(E,X)$ where $E$ is an interpolation matrix
of size  $m \times (d+1)$ and $X$ a set of $m$ knots.
Let $E_1$ be the matrix formed of the first $r+1$ columns of $E$ and $E_2$
the matrix formed of the remaining $d-r$ columns.

Suppose that $E_1$ contains at most $r+1$ 1's and $E_2$ at most $d-r$
1's.
If both pairs $(E_1,X)$, $(E_2,X)$ are regular then $(E,X)$ is regular.
\end{proposition}
\begin{proof}
The case where $|E_1|=r+1$ and $|E_2|=d-r$ is treated in Theorem~1.4
of~\cite{lorentz84}. Their argument extends to the general
case. Indeed, as shown in~\cite{lorentz84} the rank of $A(E,X)$ is at
least equal to the sum of the ranks of $A(E_1,X)$ and $A(E_2,X)$. For
the reader's convenience, we recall from~\cite{lorentz84} that this
inequality is due to the fact that $A(E,X)$ can be transformed by a
permutation of rows into a matrix of the form\footnote{We give an
  example in the appendix.}
$$\left(\begin{array}{cc}
A(E_1,X) & * \\
0 & A(E_2,X)
\end{array}\right).$$
The two matrices $A(E_1,X)$, $A(E_2,X)$ are respectively of rank $|E_1|$ and $|E_2|$ since the corresponding pairs are assumed to be regular. Thus, $A(E,X)$ is of rank at least $|E|=|E_1|+|E_2|$. This matrix must in fact be of rank exactly $|E|$ since it has $|E|$ rows, and we conclude that the pair $(E,X)$ is regular.
\end{proof}

\begin{lemma}  \label{slope}
For any finite sequence $(u_i)_{0 \leq i \leq n}$ of real numbers
with $n \geq 1$ there is an index $s \in \{0,\ldots,n-1\}$ such that
$(u_{s+t}- u_s)/t \leq (u_n - u_0)/n$ for every $t=1,\ldots,n-s$.
\end{lemma}
\begin{proof}
Let $\alpha = \min_{0 \leq i \leq n-1} (u_n - u_i) / (n-i)$ and let
$s$ be an index where the minimum is achieved.
We have $(u_n - u_s) / (n-s) = \alpha \leq (u_n - u_0)/n$.

For every $t=1,\ldots,n-s$ we also have $(u_n -u_s) / (n-s) \leq
(u_n - u_{s+t})/(n-s-t)$, which implies $(u_{s+t}-u_s)/t \leq (u_n -
u_s)/(n-s) \leq (u_n - u_0) /n$.
\end{proof}
Our lower bound results are easily derived from the following theorem.
\begin{theorem} \label{lbtool}
Consider a polynomial identity of the form:
\begin{equation} \label{identity}
\sum_{i=1}^k \alpha_i (x+x_i)^d = \sum_{i=1}^l \beta_i (x+y_i)^{e_i}
\end{equation}
where the $x_i$ are distinct real constants, the constants
$\alpha_i$ are not all zero, the $\beta_i$ and $y_i$ are arbitrary
real constants, and $e_i < d$ for every $i$. Then we must have $k+l
> (d+2)/2$.
\end{theorem}
\begin{proof}
We assume without loss of generality that the $l$ polynomials
$(x+y_i)^{e_i}$ are linearly independent. Indeed, the right-hand
side of~(\ref{identity}) could otherwise be rewritten as a linear
combination of an independent subfamily, and this would only
decrease $l$. Let us also assume that $k+l \leq (d+2)/2$. Then we
will show that the $k+l$ polynomials $(x+x_i)^d$, $(x+y_i)^{e_i}$
must be linearly independent. This is clearly in contradiction
with~(\ref{identity}).

In view of Proposition~\ref{dual}, to show that our $k+l$ polynomials
are linearly independent we need to show that the corresponding
interpolation problem has a solution space of dimension $d+1-k-l$.
Let $E$ be the corresponding interpolation matrix and $X$ the set of
knots: $|X|=m$ where $m$ is the number of distinct points
in $x_1,\ldots,x_k,y_1,\ldots,y_l$; moreover, $E$ is a matrix of size
$m \times (d+1)$ which contains $k+l$ 1's. We need to show that the
pair $(E,X)$ is regular.

Let $N_t$ be the number of 1's in the last $t$ columns of $E$. We
must have $N_1 \leq 1$, or else the independent family
$(x+y_i)^{e_i}$ would contain more than one constant polynomial. We
can now complete the proof of the theorem in the special case where
$E$ satisfies the conditions $N_t+N_{t-1} \leq t$ for every
$t=2,\ldots,d+1$. Indeed, in this case $E$ is regular by
Theorem~\ref{regular} (remember that this is how we proved
Theorem~\ref{independence}, our main linear independence result).
For the general case, the idea of the proof is to:
\begin{itemize}
\item[(i)] Split vertically $E$ in two matrices $E_1,E_2$.
\item[(ii)] Apply the same argument (Theorem~\ref{regular}) to $E_1$.
\item[(iii)] Obtain the regularity of the pair $(E_2,X)$ from the
linear independence of the $(x+y_i)^{e_i}$.
\item[(iv)] Conclude from Proposition~\ref{split} that the pair $(E,X)$ is
  regular.
\end{itemize}
We now explain how to carry out these four steps. For the first one,
note that $N_{d+1}=|E|=k+l \leq (d+2)/2$. Let us apply
Lemma~\ref{slope} to the sequence $(N_i)_{0 \leq i \leq d+1}$
beginning with $N_0=0$. The lemma shows the existence of an index $s
\in \{0,\ldots,d\}$ such that $$\frac{N_{s+t} - N_s}{t} \leq
\frac{N_{d+1}}{d+1} \leq \frac{d+2}{2(d+1)} = \frac{1}{2} +
\frac{1}{2(d+1)} \leq \frac{1}{2} + \frac{1}{2t}$$ 
for every $t=1,\ldots,d+1-s$.
Let $E_1$ be the
matrix formed of the first $r+1$ columns of $E$, where $r=d-s$. The
number of 1's in the last $t$ columns of $E_1$ is $N'_t = N_{s+t} -
N_s \leq (t+1)/2$. In particular, $N'_1 \leq 1$ and, since $N'_t,
N'_{t-1}$ are integers, we get that $N'_t + N'_{t-1} \leq \lfloor
(2t + 1)/2 \rfloor = t$ for all $t \in \{2,\ldots,r+1\}$. This
matrix therefore satisfies the hypotheses of Theorem~\ref{regular},
and we conclude that $E_1$ is regular. This completes step~(ii). For
step~(iii), we note that since $E_2$ has $s<d+1$ columns the
Birkhoff interpolation problem corresponding to the polynomials
$(x+y_i)^{e_i}$ with $e_i \leq s-1$ admits $(E_2,X)$ as its pair.
Since these polynomials are assumed to be linearly independent,
$E_2$ must contain at most $s$ 1's and $(E_2,X)$ must indeed be a
regular pair. Finally, we conclude from Proposition~\ref{split} that
$(E,X)$ is regular as well.
\end{proof}

\begin{theorem}[First lower bound] \label{lb1}
Consider a polynomial of the form
\begin{equation} \label{hardpoly1}
H_1(x)=\sum_{i=1}^k \alpha_i (x+x_i)^d
\end{equation}
 where the $x_i$ are distinct real constants, the $\alpha_i$ are
nonzero real constants, and $k \leq (d+2)/4$. If $H_1$ is written
under the form
\begin{equation} \label{easy1}
H_1(x)=\sum_{i=1}^l \beta_i (x+y_i)^{e_i}
\end{equation}
with $e_i \leq d$ for every $i$ then we must have $l \geq k$.
\end{theorem}
\begin{proof}
Assume first that $e_i < d$ for all $i$. By Theorem~\ref{lbtool} we
must have $k+l > (d+2)/2$, so $l > (d+2)/2 - k \geq k$. Consider now
the general case and assume that $l < k$. We reduce to the previous
case by moving on the side of~(\ref{hardpoly1}) the $k'$ polynomials
in~(\ref{easy1}) of degree $e_i=d$. On the second side remains a sum
of $l-k'$ terms of degree less than $d$.
 We have on the first side a sum  of terms of degree $d$.
Taking possible cancellations into account, the number of such terms
is at least $k-k'>0$, and at most $k+k'$. We must therefore have
$(k+k')+(l-k') > (d+2)/2$, so $l > (d+2)/2 -k \geq k$ after all.
\end{proof}
In other words,  writing $H_1$ under form~(\ref{hardpoly1}) is
exactly optimal when $k \leq (d+2)/4$. We can give another lower
bound of order $d$ (with an improved constant) for a polynomial of a
different form.
\begin{theorem}[Second lower bound] \label{lb2}
Let $H_2 \in \rd[X]$ be the polynomial $H_2(x)=(x+1)^{d+1}-x^{d+1}$.
If $H_2$ is written under the form
$$H_2(x)=\sum_{i=1}^l \beta_i (x+y_i)^{e_i}$$
with $e_i \leq d$ for every $i$ then we must have $l > (d-1)/2$.
\end{theorem}
\begin{proof}
This follows directly from Theorem~\ref{lbtool} after replacing $d$
by $d+1$ in~(\ref{identity}). Since $k=2$ we must have
$2+l>(d+3)/2$, i.e., $l > (d-1)/2$.
\end{proof}

This result shows that allowing exponents $e_i >d$ can drastically
decrease the ``complexity'' of a polynomial since $H_2$ can be expressed
as the difference of only two $(d+1)$-st powers.
Such savings cannot be obtained for all polynomials. Indeed, the next
result, which subsumes Theorem~\ref{lb1}, shows that no improvement is
possible for $H_1$ even if arbitrarily large powers are allowed.
\begin{theorem}[Third lower bound]
Consider a polynomial of the form
\begin{equation} 
H_1(x)=\sum_{i=1}^k \alpha_i (x+x_i)^d
\end{equation}
 where the $x_i$ are distinct real constants, the $\alpha_i$ are
nonzero real constants, and $k \leq (d+2)/4$. If $H_1$ is written
under the form
\begin{equation} \label{easy3}
H_1(x)=\sum_{i=1}^l \beta_i (x+y_i)^{e_i}
\end{equation}
then we must have $l \geq k$.
\end{theorem}
Note that the exponents $e_i$ may be arbitrarily large.
\begin{proof}
Let $n$ be the largest exponent $e_i$ which occurs with a
coefficient $\beta_i \neq 0$. The case $n \leq d$ is covered by
Theorem~\ref{lb1}, so we assume here that $n>d$. In
equation~(\ref{easy3}), let us move all the $n$-th powers from the
right hand side to the left hand side, and the $k$ degree-$d$ terms
of $H_1$ from the left hand side to the right hand side. Applying
Theorem~\ref{lbtool} to this identity shows that $k+l>(n+2)/2$.
Hence $l > (n+2)/2 - k \geq k$.
\end{proof}

\begin{remark} \label{opt}
The lower bound for $H_2$ in Theorem~\ref{lb2} is essentially
optimal. More concretely, it is optimal up to one unit when $d$ is
even, and exactly optimal when $d$ is odd.

Note indeed that by a change of variable, representing $H_2$ is
equivalent to representing the polynomial
$H_3(x)=(x+1)^{d+1}-(x-1)^{d+1}$.
If we expand the  two binomials in $H_3$  the monomials of degree
$d+1-j$ wih even $j$ cancel,
and we obtain a sum of $\lceil \frac{d+1}{2} \rceil$ monomials.
See Proposition~\ref{dependenceincomplex} for a generalization of this
observation. In fact, with the same number of terms we can represent not only
$H_2$ but all polynomials of degree $d$: see Proposition~\ref{upper} below.

The consideration of $H_3$ also shows that Theorem~\ref{independence}
is 
optimal up to one unit when $d$ is
even, and exactly optimal when $d$ is odd.
Indeed, we have just observed that there is a
linear dependence between the $2+\lceil \frac{d+1}{2} \rceil$
polynomials $(x+1)^{d+1},(x-1)^{d+1},x^d,x^{d-2},x^{d-4},\ldots$.

If $d$ is odd, the number of polynomials of degree less than $j$ in this sequence is $n_j=\lfloor j/2 \rfloor$ for $j \leq d+1$;
moreover, $n_{d+2}=2+(d+1)/2=(d+5)/2$.
Hence $n_j+n_{j+1}=j$ for $j \leq d$;
moreover, $n_{d+1}+n_{d+2}=d+3$.

If $d$ is even, the number of polynomials of degree less than $j$ in this sequence is $n_j=\lceil j/2 \rceil$ for $j \leq d+1$;
moreover, $n_{d+2}=2+(d+2)/2=(d+6)/2$.
Hence $n_j+n_{j+1}=j+1$ for $j \leq d$;
moreover,  $n_{d+1}+n_{d+2}=d+4$.
\end{remark}
A simple construction shows that all polynomials of degree $d$ can be
written as a linear combination of $\lceil(d+1)/2 \rceil$ powers.
\begin{proposition} \label{upper}
Every polynomial of degree $d$ can be
expressed as $\sum_{i = 1}^l \beta_i (x + y_i)^{e_i}$ with $l \leq
\lceil(d+1)/2 \rceil$.
\end{proposition}
\begin{proof}
We use induction on $d$.
Since the result is obvious for  $d = 0,1$
we consider a polynomial $f = \sum_{i = 0}^d a_i x^i$ of degree $d
\geq 2$, and we assume that that the Proposition holds for polynomials
of degree $d-2$.
We observe that $g := f - a_d (x + (a_{d-1}/d a_d))^d$ has degree
$\leq d-2$. Applying the induction hypothesis to $g$ we get that $g
= \sum_{i = 1}^{l'} \beta_i (x + y_i)^{e_i}$, with $l' \leq \lceil(d-1)/2 \rceil$.
Hence, setting $l = l'+1$, $\beta_l = a_d$,  $y_l = a_{d-1}/(d a_d)$
and $e_l = d$, we conclude that $f = \sum_{i = 1}^l \beta_i (x +
y_i)^{e_i}$ and $l \leq 1 + \lceil(d-1)/2 \rceil = \lceil(d+1)/2 \rceil$.
\end{proof}
Theorem~\ref{lb2} therefore shows that of all polynomials of degree
$d$, $H_2$ is essentially (up to a small additive constant) the
hardest one.

\section{Changing Fields} \label{fields}

Some of the proof techniques used in this paper,
and even the results themselves,
are specific to the field of real numbers. This is due to the fact
that certain linear dependence relations which cannot occur over $\rr$
may occur if we change the base  field.
For instance, over a field of characteristic $p>0$ we have
$(X+1)^{p^k}-X^{p^k}-1=0$ for any $k$ (compare with
Theorem~\ref{independence} for the real case).
The remainder of this section is devoted to a discussion of the
complex case. We begin with an identity which generalizes the identity
$(x+1)^2-(x-1)^2-4x=0$.
\begin{proposition}\label{dependenceincomplex}
Take $k \in \mathbb{Z}^+$ and let $\xi$ be a $k$-th primitive root
of unity. Then, for all $d \in \mathbb{Z}^+$ and all $\mu \in
\mathbb{C}$ the following equality holds:
$$\sum_{j = 1}^{k} \xi^j (x + \xi^j \mu)^d = \sum_{i \equiv -1\, ({\rm mod}\ k) \atop 0 \leq i \leq d} k \binom{d}{i} \mu^i x^{d-i}.$$
\end{proposition}

\begin{proof}
We observe that $$\sum_{j = 1}^{k} \xi^j (x + \xi^j \mu)^d = \sum_{i
= 0}^d \binom{d}{i} \mu^i x^{d-i} \left(\sum_{j = 1}^{k} \xi^{ji +
j}\right).$$ To deduce the result it suffices to prove that $\sum_{j
= 1}^{k} \xi^{ji + j}$ equals $k$ if $i \equiv -1 \ ({\rm mod} \
k)$, or $0$ otherwise. Whenever $i \equiv -1 \ ({\rm mod}\ k)$ we
have that $\xi^{ji + j} = (\xi^{i+1})^j = 1$ for all $j \in
\{1,\ldots,k\}$. For $i \not\equiv -1\ ({\rm mod}\ k)$, the
summation of the geometric series shows that
$$\displaystyle \sum_{j =1}^{k} \xi^{j(i + 1)} =
\xi^{i + 1}.\frac{\xi^{k(i + 1)-1}}{\xi^{i + 1}-1}=0.$$
\end{proof}
For any $d, k \in \mathbb{Z}^+$,
Proposition~\ref{dependenceincomplex} yields an identity of the form
\begin{equation} \label{complexid}
\sum_{j = 1}^{k} \alpha_j (x + x_j)^d = \sum_{j = 1}^l \beta_j
x^{e_j}
\end{equation} where the $x_j$
are distinct complex constants, the $\alpha_j, \beta_j$ are nonzero
complex numbers, $l=\big\lfloor \frac{d+1}{k}\big\rfloor$
and  $e_j < d$ for all $j$.
Note the sharp contrast with theorems~\ref{lbtool} and~\ref{lb1}.
In particular, Theorem~\ref{lb1} gives an $\Omega(d)$ lower bound for
polynomials of the form $\sum_{j = 1}^{k} \alpha_i (x + x_i)^d$ over
the field of real numbers (the implied constant in the $\Omega$
notation is equal to 1/4).
But in~(\ref{complexid}) we have $k.l \leq d+1$, so $k \leq
\sqrt{d+1}$ or $l \leq \sqrt{d+1}$. We conclude that no better lower
bound than $\Omega(\sqrt{d})$ can possibly hold over $\mathbb{C}$ for the
same family of polynomials, at least for arbitrary distinct $x_i$'s
and arbitrary nonzero $\alpha_i$.
Such a $\Omega(\sqrt{d})$
lower bound was recently obtained for the more general model of
sums of power of bounded degree polynomials: see Theorem~2 in~\cite{KKPS}.

We leave it as an open problem to close this quadratic gap between
lower bounds over $\rr$ and $\mathbb{C}$:
find an explicit polynomial $f \in \mathbb{C}[X]$ of degree $d$
which requires at
least $k=\Omega(d)$ terms to be represented under the form
$$f(x)=\sum_{i=1}^k \alpha_i (x+x_i)^{e_i}.$$
With the additional requirements $e_i \leq d$ for all $i$, the ``target polynomial'' $H_2(x)=(x+1)^{d+1}-x^{d+1}$ from
Theorem~\ref{lb2} looks like a plausible candidate.

{\small

\section*{Acknowledgments}

P.K. acknowledges useful discussions with member of the Aric research group in the initial stages of this work.



\appendix

\section*{Appendix: a worked out example}

We illustrate the proof method of Theorem~\ref{lbtool} (more than
the result itself) on a small example: we show with this method
that there is no
identity of the form
\begin{equation} \label{example}
\alpha_1x^5+\alpha_2(x+1)^5+\alpha_3(x+3)^5=\beta_1 x^2 + \beta_2
(x+1)+\beta_3(x+3)^2
\end{equation}
except if the coefficients $\alpha_i,\beta_i \in \rr$ are all 0.
The corresponding interpolation problem is:
\begin{equation} \label{interpol}
g(0)=0,\ g(1)=0,\ g(3)=0,\ g^{(3)}(0)=0,\ g^{(4)}(1)=0,\ g^{(3)}(3)=0
\end{equation}
where $g \in {\rr}_5[X]$.
The set of knots is $X=\{x_1,x_2,x_3\}=\{0,1,3\}$ and the
interpolation matrix is
$$E=\left(\begin{array}{cccccc}
1 & 0 & 0 & 1 & 0 & 0\\
1 & 0 & 0 & 0 & 1 & 0\\
1 & 0 & 0 & 1 & 0 & 0
\end{array}\right).$$
This matrix is not order regular.
Indeed, the pair $(E,Y)$ where $Y=\{-1,0,1\}$ is not regular.
This follows from the identity $$(x+1)^2-(x-1)^2-4x=0$$ which was pointed
out earlier in the paper. We will nonetheless show that the pair
$(E,X)$ is regular. Toward this, let us split $E$ in the middle to obtain
the two matrices
$$E_1=\left(\begin{array}{ccc}
1 & 0 & 0 \\
1 & 0 & 0 \\
1 & 0 & 0
\end{array}\right)$$
and
$$E_2=\left(\begin{array}{ccc}
1 & 0 & 0 \\
0 & 1 & 0 \\
1 & 0 & 0
\end{array}\right).$$
The first matrix is regular since all its 1's are in the first column.
The second matrix fails to be order regular for the same reason that
$E$ does, but it is easy to check that the
interpolation
problem $h(0)=0,\ h'(1)=0,\ h(3)=0$ has no nontrivial solution
in ${\rr}_2[X]$. Hence the pair $(E_2,X)$ is regular.
It follows from Proposition~\ref{split} that $(E,X)$ is a regular
pair, and the  6 polynomials in~(\ref{example}) are indeed linearly independent.

We conclude with a remark about Proposition~\ref{split}.
In the proof of this result, we used the fact that
the matrix $A(E,X)$ of the linear system can be transformed by a
permutation of rows into a matrix of the form
$$\left(\begin{array}{cc}
A(E_1,X) & * \\
0 & A(E_2,X)
\end{array}\right).$$
We point out that when the 6 interpolation constraints are
listed in the same order as in~(\ref{interpol}), $A(E,X)$ is already
in this form.
In particular, the equations for the last 3 constraints are:
$$g_3+x_1g_4+x_1^2g_5/2=0,$$
$$g_4+x_2g_5=0,$$
$$g_3+x_3g_4+x_3^2g_5/2=0$$
and the matrix of this subsystem is indeed $A(E_2,X)$.
As to $A(E_1,X)$, consider for instance the third constraint $g(x_3)=0$.
The corresponding equation is $\sum_{j=0}^5 x_3^jg_j/j! = 0$.
The first 3 coefficients are $1,x_3,x_3^2/2$ and this is the last row
of $A(E_1,X)$.
\end{document}